\documentclass[letterpaper,accepted=2019-08-02,onecolumn]{quantumarticle}
\pdfoutput=1

\usepackage{amsfonts}
\usepackage{amsmath}
\usepackage{graphicx}
\usepackage{latexsym}
\usepackage[sorting=none,maxbibnames=99]{biblatex}
\renewbibmacro{in:}{}
\bibliography{bib}

\usepackage[margin=1.05in]{geometry}

\usepackage[T1]{fontenc}
\usepackage{color,graphicx}
\usepackage{array}
\usepackage{enumerate}
\usepackage{amsmath}
\usepackage{amssymb}
\usepackage{amsthm}
\usepackage{pgfplots}
\usepackage{pgf}
\usepackage{tikz}
\usepackage{filecontents}
\usepgfplotslibrary{patchplots} 
\usetikzlibrary{pgfplots.patchplots} 
\pgfplotsset{width=9cm,compat=1.5.1}

\usepackage{amsthm}

\newtheorem{definition}{Definition}
\newtheorem{lemma}{Lemma}
\newtheorem{theorem}{Theorem}

\newtheorem{proposition}{Proposition}

\newtheorem{example}{Example}

\usepackage{xcolor}
\usepackage{makeidx}
\usepackage[colorlinks=true,linkcolor=blue,anchorcolor=blue,citecolor=red,urlcolor=magenta]{hyperref}

\newcommand{\tr}{\operatorname{Tr}}
\newcommand{\rank}{\operatorname{rank}}

\newcommand{\bra}[1]{\langle #1 |}
\newcommand{\ket}[1]{| #1 \rangle}

\newcommand{\ketbra}[2]{| #1 \rangle\langle #2 |}

\newcommand\Tr{\mathop{\rm Tr}\nolimits}
\newcommand\supp{\mathop{\rm supp}\nolimits}

\newcommand{\defeq}{\stackrel{\smash{\textnormal{\tiny def}}}{=}}

\def\C{\mathbb{C}}

\def\R{\mathbb{R}}

\def\D{\mathcal{D}}
\def\Lin{\mathcal{L}}
\def\HP{\mathcal{L}^{\textup{H}}}

\def\S{\mathcal{S}}

\def\0{\mathbf{0}}

\def\Pos{\operatorname{Pos}}

\begin{document}

\title{The Non-m-Positive Dimension of a Positive Linear Map}

\author{Nathaniel Johnston}\affiliation{Department of Mathematics \& Computer Science, Mount Allison University, Sackville, NB, Canada E4L 1E4}\affiliation{Department of Mathematics \& Statistics, University of Guelph, Guelph, ON, Canada N1G 2W1}
	
\author{Benjamin Lovitz}\affiliation{Institute for Quantum Computing, Department of Applied Mathematics, University of Waterloo, Waterloo, ON, Canada N2L~3G1}

\author{Daniel Puzzuoli}\affiliation{Department of Mathematics and Statistics, University of Ottawa, Ottawa, ON, Canada K1N 6N5}\affiliation{School of Mathematics and Statistics, Carleton University, Ottawa, ON, Canada K1S 5B6}

\date{August 2, 2019}

\maketitle

\begin{abstract}
	We introduce a property of a matrix-valued linear map $\Phi$ that we call its ``non-m-positive dimension'' (or ``non-mP dimension'' for short), which measures how large a subspace can be if every quantum state supported on the subspace is non-positive under the action of $I_m \otimes \Phi$. Equivalently, the non-mP dimension of $\Phi$ tells us the maximal number of negative eigenvalues that the adjoint map $I_m \otimes \Phi^*$ can produce from a positive semidefinite input. We explore the basic properties of this quantity and show that it can be thought of as a measure of how good $\Phi$ is at detecting entanglement in quantum states. We derive non-trivial bounds for this quantity for some well-known positive maps of interest, including the transpose map, reduction map, Choi map, and Breuer--Hall map. We also extend some of our results to the case of higher Schmidt number as well as the multipartite case. In particular, we construct the largest possible multipartite subspace with the property that every state supported on that subspace has non-positive partial transpose across at least one bipartite cut, and we use our results to construct multipartite decomposable entanglement witnesses with the maximum number of negative eigenvalues.
\end{abstract}

\section{Introduction}

One of the central problems in quantum information theory is that of determining whether a given quantum state is separable or entangled \cite{GT09,HHH09}. One of the standard approaches to this problem is to make use of entanglement witnesses (or equivalently, positive matrix-valued linear maps)---Hermitian operators that verify the entanglement in certain subsets of states.

While there are several methods of quantifying the amount of entanglement in a quantum state (see \cite{PV07,Vid00} and the references therein), there are hardly any methods of quantifying the effectiveness of an entanglement witness (or equivalently, of a positive linear map) at detecting entanglement in quantum states. The only notion along these lines that we are aware of is that of one entanglement witness being ``finer'' than another \cite{LKCH00}, which refers to the situation where one entanglement witness detects entanglement in a superset of quantum states that are detected by the other.

This paper aims to fill this gap somewhat by introducing what we call the ``non-$m$-positive dimension'' and ``non-completely-positive ratio'' as measures of how effective a positive linear map (or entanglement witness) is at detecting entanglement. Given a positive linear map $\Phi$ on $M_n$, these measures quantify how large of a subspace of $\C^m \otimes \C^n$ the augmented linear map $I_m \otimes \Phi$ detects entanglement in. Equivalently, this quantity tells us how many negative eigenvalues the adjoint map $I_m \otimes \Phi^*$ can produce when applied to a quantum state (with at least one negative eigenvalue meaning that the entanglement in that state was detected).

The paper is organized as follows. In Section~\ref{sec:prelims}, we review the mathematical preliminaries of quantum entanglement detection that form the basis of this paper. In Section~\ref{sec:non_mp_dim}, we introduce the non-$m$-positive dimension and non-completely-positive ratio, explore some of their basic properties, and come up with a general non-trivial upper bound on them that can be efficiently computed. In Section~\ref{sec:specific_maps}, we apply our bounds and methods to several well-known positive linear maps, including the transpose map, reduction map, Choi map, and Breuer--Hall map. In Section~\ref{sec:relationship_other_measures}, we briefly show that the non-completely-positive ratio agrees with the notion of ``finer'' entanglement witnesses. In Section~\ref{sec:schmidt_k_pos}, we briefly discuss how to extend our results to the case of $k$-positive linear maps (which detect large Schmidt number in quantum states). In Section~\ref{sec:multipartite_ppt}, we solve the multipartite version of this problem for the transpose map, thus generalizing known bipartite results about the partial transpose, including (for example) how many negative eigenvalues decomposable entanglement witnesses can have. Finally, we close in Section~\ref{sec:conclusions} with some open questions and directions for future research.

\section{Mathematical Preliminaries}\label{sec:prelims}

Our notation and terminology is quite standard, so our introduction here is brief---for a more thorough treatment of the mathematics of quantum entanglement and matrix-valued linear maps, we direct the reader to a survey paper like \cite{GT09,HHH09} or a textbook like \cite{NC00,Wat18}.

We use $M_n$ to denote the set of $n \times n$ complex matrices and $\Pos(M_n)$ to denote the set of (Hermitian) positive semidefinite matrices in $M_n$. The ordering $A \succeq B$ refers to the Loewner ordering---$A \succeq B$ means that $A-B$ is positive semidefinite, and in particular $A \succeq O$ means that $A$ is positive semidefinite. We use $\D(M_n) \subseteq \Pos(M_n)$ to denote the set of density matrices (mixed quantum states), which are positive semidefinite matrices with trace $1$. We typically denote density matrices by lowercase Greek letters like $\rho \in \D(M_n)$.

We say that a linear map $\Phi$ acting on $M_n$ is \emph{Hermiticity-preserving} if $\Phi(X)$ is Hermitian whenever $X$ is Hermitian, and we denote the set of such maps by $\HP(M_n,M_n)$. We say that $\Phi \in \HP(M_n,M_n)$ is \emph{positive} if $\Phi(X) \in \Pos(M_n)$ whenever $X \in \Pos(M_n)$. More generally, we say that $\Phi$ is \emph{$m$-positive} if $I_m \otimes \Phi$ is positive, and we say that $\Phi$ is \emph{completely positive} if it is $m$-positive for all $m \geq 1$. The \emph{adjoint} of a linear map $\Phi \in \HP(M_n,M_n)$ is the map $\Phi^* \in \HP(M_n,M_n)$ defined by $\tr(X\Phi(Y)) = \tr(\Phi^*(X)Y)$ for all Hermitian $X,Y \in M_n$, and it is straightforward to show that $\Phi$ is $m$-positive if and only if $\Phi^*$ is $m$-positive.

Given a linear map $\Phi \in \HP(M_n,M_m)$, its \emph{Choi matrix} is the (Hermitian) matrix
\begin{align*}
	J(\Phi) \defeq n(I_n \otimes \Phi)\big(\ketbra{\psi^+}{\psi^+}\big) \in M_n \otimes M_m,
\end{align*}
where $\ket{\psi^+} = \frac{1}{\sqrt{n}}\sum_{i=1}^n\ket{i}\otimes\ket{i} \in \C^n \otimes \C^n$ is the standard maximally-entangled pure state (and $\{\ket{i}\}$ denotes the standard basis of $\C^n$). It is well-known that $\Phi$ is completely positive if and only if $J(\Phi)$ is positive semidefinite, if and only if it is $n$-positive \cite{Cho75}.

A density matrix $\rho \in \D(M_m \otimes M_n)$ is called \emph{separable} \cite{Wer89} if there exist $\{X_i\} \subseteq \Pos(M_m)$ and $\{Y_i\} \subseteq \Pos(M_n)$ such that $\rho = \sum_i X_i \otimes Y_i$, and it is called \emph{entangled} otherwise. Positive linear maps characterize separability in the sense that $\rho$ is separable if and only if $(I_m \otimes \Phi)(\rho) \succeq O$ for all positive maps $\Phi \in \HP(M_n,M_m)$. Equivalently, $\rho$ is separable if and only if $\tr(J(\Phi)\rho) \geq 0$ for all positive maps $\Phi \in \HP(M_m,M_n)$. For this reason, the Choi matrix $J(\Phi)$ of a positive but not completely positive matrix is called an \emph{entanglement witness}. If $\Phi$ is positive and $(I_m \otimes \Phi)(\rho) \not\succeq O$ then we say that $\Phi$ ``detects'' the entanglement in $\rho$, and similarly if $\tr(J(\Phi)\rho) < 0$ then we say that $J(\Phi)$ ``detects'' the entanglement in $\rho$.

The standard example of a map that is positive but not completely positive (and not even $2$-positive) is the transpose map $T \in \HP(M_n,M_n)$---see \cite{Pau03}, for example. The map $I_m \otimes T$ is called the \emph{partial transpose}, and states $\rho \in \D(M_m \otimes M_n)$ for which $(I_m \otimes T)(\rho) \succeq O$ are called \emph{positive partial transpose} (or \emph{PPT} for short). Our discussion of separable states above tells us that every separable state is PPT, but the converse does not hold.

The \emph{trace norm} of a matrix $X \in M_n$, denoted by $\|X\|_1$, is the sum of the singular values of $X$. The \emph{induced trace norm} \cite{Wat05} of a linear map $\Phi \in \HP(M_n,M_n)$ is the quantity
\[
    \|\Phi\|_1 \defeq \max \left\{ \big\|\Phi(X)\big\|_1 : X \in M_n , \|X\|_1 \leq 1\right\}.
\]
This norm is not stable under tensor products (unless $\Phi$ is completely positive, in which case we have $\|I_m \otimes \Phi\|_1 = \|\Phi^*(I)\|$ for all $m \geq 1$, where $\|\cdot\|$ denotes the usual operator norm), so we define $\|\Phi\|_m = \|I_m \otimes \Phi\|_1$ for all $m \geq 1$. We similarly define $d_m(\Phi_1,\Phi_2) \defeq \|\Phi_1 - \Phi_2\|_m$ to be the distance between two maps $\Phi_1,\Phi_2 \in \HP(M_n,M_n)$ in this norm, and we note that $\|\Phi\|_m = \|\Phi\|_n$ whenever $m \geq n$. In this case, this norm is typically referred to as the \emph{diamond norm} of $\Phi$ and denoted by $\|\Phi\|_{\diamond}$ \cite{Kit97}, and we similarly use the notation $d_\diamond(\Phi_1,\Phi_2) = d_m(\Phi_1,\Phi_2)$ whenever $m \geq n$.

\section{The Non-Positive Dimension and Ratio}\label{sec:non_mp_dim}

We now introduce the basic properties of the central quantity of interest in this paper.

\begin{definition}\label{defn:non_p_dim}
    If $m \geq 1$ is an integer and $\Phi \in \HP(M_n,M_n)$ then we say that the \textbf{non-$\mathbf{m}$-positive dimension} (or \textbf{non-$\mathbf{m}$P dimension}) of $\Phi$, denoted by $\nu_m(\Phi)$, is the quantity
    \begin{align*}
        \nu_m(\Phi) \defeq \max_{\S \subseteq \C^m \otimes \C^n} \big\{ \dim(\S) : & \ \S \ \text{is a subspace such that} \ (I_m \otimes \Phi)(\rho) \not\succeq O \ \text{whenever} \\
        & \ \rho \in \D(M_m \otimes M_n) \ \text{and} \ \supp(\rho) \subseteq \mathcal{S} \big\}.
    \end{align*}
\end{definition}

If we recall that the statement $(I_m \otimes \Phi)(\rho) \not\succeq O$ means that $\Phi$ detects entanglement in the state $\rho \in \D(M_m \otimes M_n)$, then it seems natural to think of $\nu_m(\Phi)$ as measuring how much space $\Phi$ is capable of detecting entanglement in. We thus think of $\nu_m(\Phi)$ as a rough measure of how ``good'' $\Phi$ is at detecting entanglement, with higher values corresponding to ``better'' entanglement detection.

The following theorem provides us with another way of thinking about $\nu_m(\Phi)$ in terms of its adjoint map $\Phi^*$, which has a similar interpretation of ``how much entanglement $\Phi^*$ can detect'').

\begin{theorem}\label{thm:neg_evals_subspace_corr}
    Suppose $\Phi \in \HP(M_n,M_n)$. Then the maximum number of strictly negative eigenvalues that $(I_m \otimes \Phi^*)(\rho)$ can have as $\rho$ ranges over $\D(M_m \otimes M_n)$ is exactly $\nu_m(\Phi)$.
\end{theorem}

The correspondence between $\nu_m(\Phi)$ and the number of negative eigenvalues of $(I_m \otimes \Phi^*)(\rho)$ was made explicit in \cite{Joh13} in the case when $\Phi = T$ is the transpose map, and the proof of this more general claim follows the same logic in a straightforward manner. We prove it explicitly for completeness.

\begin{proof}[Proof of Theorem~\ref{thm:neg_evals_subspace_corr}]
    Throughout this proof, let $p$ denote the maximum number of negative eigenvalues that $(I_m \otimes \Phi^*)(\rho)$ can have.
    
    To see that $\nu_m(\Phi) \geq p$, let $\rho \in \D(M_m \otimes M_n)$ be such that $(I_m \otimes \Phi^*)(\rho)$ has $p$ strictly negative eigenvalues. Write $(I_m \otimes \Phi^*)(\rho) = P_{+} - P_{-}$ where $P_{+},P_{-}$ are positive semidefinite with orthogonal support (i.e., they come from the spectral decomposition of $(I_m \otimes \Phi^*)(\rho)$). If we let $\mathcal{S}$ be the support of $P_{-}$ then $\dim(\mathcal{S}) = p$, and we claim that every state $\sigma \in \D(M_m \otimes M_n)$ with support in $\mathcal{S}$ is such that $(I_m \otimes \Phi)(\sigma) \not\succeq O$. To see why this claim holds, we note that
    \begin{align*}
        \tr\big((I_m \otimes \Phi)(\sigma)\rho\big) = \tr\big(\sigma(I_m \otimes \Phi^*)(\rho)\big) = -\tr(\sigma P_{-}) < 0,
    \end{align*}
    so $(I_m \otimes \Phi)(\sigma) \not\succeq O$, as claimed. It follows that $\nu_m(\Phi) \geq p$.
    
    Conversely, to see that $p \geq \nu_m(\Phi)$, let $P \in \Pos(M_m \otimes M_n)$ be the orthogonal projection onto some subspace $\mathcal{S} \subseteq \C^m \otimes \C^n$ for attaining the maximum in the definition of $\nu_m(\Phi)$ (i.e., $\dim(\mathcal{S}) = \nu_m(\Phi)$). Since there is no $\sigma \in \D(M_m \otimes M_n)$ satisfying both $(I_m \otimes \Phi)(\sigma) \succeq O$ and $P\sigma P = \sigma$, we have $\Tr(P\sigma) < 1$ whenever $(I_m \otimes \Phi)(\sigma) \succeq O$. Furthermore, since the set of states satisfying $(I_m \otimes \Phi)(\sigma) \succeq O$ is compact, there exists a real constant $0 < c < 1$ such that $\Tr(P\sigma) \leq c$ for all $\sigma \in \D(M_m \otimes M_n)$ satisfying $(I_m \otimes \Phi)(\sigma) \succeq O$. If we define the matrix $X = I - \tfrac{1}{c}P$ then $X$ has exactly $\dim(\mathcal{S}) = \nu_m(\Phi)$ negative eigenvalues and $\Tr(X\sigma) \geq 0$ for all $\sigma \in \D(M_m \otimes M_n)$ with $(I_m \otimes \Phi)(\sigma) \succeq O$. The latter fact is equivalent to the statement that $X$ is in the \emph{dual cone} $\mathcal{C}^\circ$ of the set
    \[
        \mathcal{C} = \big\{ \sigma \in \D(M_m \otimes M_n) : (I_m \otimes \Phi)(\sigma) \succeq O \big\}.
    \]
    
    By using basic facts about dual cones (see \cite{BV04,SSZ09}, for example), we see that there exist $X_1,X_2 \in \Pos(M_m \otimes M_n)$ such that $X = X_1 + (I_m \otimes \Phi^*)(X_2)$. Finally, we define $\rho = X_2/\Tr(X_2) \in \D(M_m \otimes M_n)$. Since $X$ has $\nu_m(\Phi)$ negative eigenvalues, and (up to scaling) $(I_m \otimes \Phi^*)(X_2) = X - X_1$, it follows that $(I_m \otimes \Phi^*)(\rho)$ has at least $\nu_m(\Phi)$ negative eigenvalues as well. That is, we have shown that $p \geq \nu_m(\Phi)$.
\end{proof}

Before proceeding to discuss the basic properties of $\nu_m(\Phi)$ and compute bounds on it, it is important to emphasize that it does not stabilize in the same way that $m$-posivity of linear maps does. While it is true that $\Phi \in \HP(M_n,M_n)$ being $n$-positive implies that it is $m$-positive for all $m \geq n$ (i.e., $\nu_n(\Phi) = 0$ implies $\nu_m(\Phi) = 0$ for all $m \geq n$), it is \emph{not} the case that $\nu_m(\Phi) = \nu_n(\Phi)$ for all $m \geq n$ in general.

For example, it is known that the transpose map $T \in \HP(M_n,M_n)$ satisfies $\nu_m(T) = (m-1)(n-1)$ for all $m$ and $n$ \cite{Joh13}, so $\nu_m(T) \rightarrow \infty$ as $m \rightarrow \infty$. Something similar actually happens for every non-completely-positive map $\Phi \in \HP(M_n,M_n)$: if $(I_m \otimes \Phi^*)(\rho)$ has $\nu_m(\Phi) \geq 1$ negative eigenvalues, then for any integer $k \geq 1$ we can define a state $\rho_k \in \D(M_{km} \otimes M_n)$ by
\[
    \rho_k = \frac{1}{k}\sum_{i=0}^{k-1} (P_i \otimes I_n)\rho(P_i^* \otimes I_n),
\]
where each $P_i \in M_{mk,m}$ is defined by $P_i\ket{j} = \ket{j+mi}$ for all $0 \leq j < m$. It is then straightforward to check that $(I_{km} \otimes \Phi^*)(\rho_k)$ has $k\nu_m(\Phi) \geq k$ negative eigenvalues, so necessarily $\nu_{km}(\Phi) \rightarrow \infty$ as $k \rightarrow \infty$. More generally, this argument implies that $\nu_{km}(\Phi) \geq k \nu_m(\Phi)$ for any $\Phi$, $k$, and $m$. This leads somewhat naturally to the following definition, which can be thought of as a regularization of $\nu_m(\Phi)$ over $m$:

\begin{definition}\label{defn:non_cp_dim}
    If $\Phi \in \HP(M_n,M_n)$ then we say that the \textbf{non-completely-positive ratio} (or \textbf{non-CP ratio}) of $\Phi$, denoted by $\nu(\Phi)$, is the quantity
    \begin{align*}
        \nu(\Phi) \defeq \lim_{m\rightarrow\infty} \frac{\nu_m(\Phi)}{m}.
    \end{align*}
\end{definition}

Before proceeding to prove some basic facts and bounds on these quantities, we should clarify that, although $\nu_m(\Phi)$ is always an integer, $\nu(\Phi)$ need not be. Also note that we have implicitly assumed that the limit in the above definition exists. We will give the proof that this limit always exists in the next section.

\subsection{Bounds and Basic Properties}\label{sec:bounds_properties}

We now list some basic properties of $\nu_m(\Phi)$ and $\nu(\Phi)$ that can each be proved in a line or two, or that follow immediately from known results. We have already mentioned a few of these properties earlier, but we list them here anyway for ease of reference. In each of the following statements, $\Phi \in \HP(M_n,M_n)$ is positive.

\begin{itemize}
    \item For all $m \geq k \geq 1$, we have $\nu_m(\Phi) \geq \nu_k(\Phi)$.
    
    \item $\Phi$ is $m$-positive if and only if $\nu_m(\Phi) = 0$. In particular, $\Phi$ is completely positive if and only if $\nu_n(\Phi) = 0$, if and only if $\nu(\Phi) = 0$. In other words, if $\Phi$ is capable of detecting entanglement in some state then $\nu_n(\Phi) \geq 1$ and $\nu(\Phi) \geq 1/n$.
    
    \item If $\Phi$ is not $m$-positive then $\nu(\Phi) \geq 1/m$. This follows from the fact that $\nu_m(\Phi) \geq 1$ along with the argument from earlier that $\nu_{km}(\Phi) \geq k\nu_m(\Phi)$.
    
    \item $\nu_m(\Phi) \leq (m-1)(n-1)$ and thus $\nu(\Phi) \leq n-1$. This follows from the fact that every subspace of dimension larger than $(m-1)(n-1)$ contains a separable state $\ket{v} \otimes \ket{w}$ \cite{Par04}, and $(I_m \otimes \Phi)(\ketbra{v}{v} \otimes \ketbra{w}{w}) = \ketbra{v}{v} \otimes \Phi(\ketbra{w}{w}) \succeq O$.
    
    \item More generally, if $\Phi$ is $k$-positive then $\nu_m(\Phi) \leq (m-k)(n-k)$ and thus $\nu(\Phi) \leq n-k$. This follows from the fact that every subspace of dimension larger than $(m-k)(n-k)$ contains a state with Schmidt rank at most $k$ \cite{CW08}, and $I_m \otimes \Phi$ sends such states to PSD operators \cite{TH00}.
    
    \item The transpose map $T \in \HP(M_n,M_n)$ attains the bound $\nu_m(T) = (m-1)(n-1)$ (and thus the bound $\nu(T) = n-1$) for all $m$ and $n$ \cite{Joh13}.
    
     \item If $\Psi \in \HP(M_n,M_n)$ is $m$-positive then $\nu_m(\Phi + \Psi) \leq \nu_m(\Phi)$. In particular, if $\Psi$ is completely positive then $\nu(\Phi + \Psi) \leq \nu(\Phi)$.
    
    \item If $\Psi \in \HP(M_n,M_n)$ is $m$-positive then $\nu_m(\Psi \circ \Phi) \leq \nu_m(\Phi)$. In particular, if $\Psi$ is completely positive then $\nu( \Psi \circ \Phi) \leq \nu(\Phi)$.
    
    \item If $\Psi \in \HP(M_n,M_n)$ is $m$-positive, the inequality $\nu_m(\Phi \circ \Psi) \leq \nu_m(\Phi)$ does not necessarily hold, as we now demonstrate. Let $n=2k$ be even and let $\Psi^*$ be defined by $\Psi^*(X)=I_k \otimes X_2$, where $X_2$ is the upper-left $2 \times 2$ block of $X$. Let $\Phi^*$ be defined by $\Phi^*(X)=\ketbra{0}{0}\otimes X_2^T$. Then $\Psi^*$ (and therefore $\Psi$) is completely positive, $\Phi$ is positive, $\nu_m(\Phi)=m-1$, and $\nu_m(\Phi \circ \Psi)=k(m-1).$ Letting $n\geq 4$ and $m=n$ gives $\nu_n(\Phi \circ \Psi) > \nu_n(\Phi)$.
    
    \item The non-$m$-positive dimension of a positive map is not necessarily equal to that of its adjoint. Indeed, for $\Psi$ as in the above bullet, $\nu_m(\Phi \circ \Psi)=k(m-1)$, but \begin{align*}
{\nu_m((\Phi \circ \Psi)^*)=\nu_m(\Psi^* \circ \Phi^*)\leq \nu_m(\Phi^*)=\nu_m(\Phi)=m-1}
\end{align*}
(and in fact, $\nu_m(\Psi^* \circ \Phi^*)=m-1$).

\end{itemize}

\begin{proposition}\label{prop:limit_exists}
For any $\Phi \in \HP(M_n,M_n)$, the limit in Definition \ref{defn:non_cp_dim} exists.
\begin{proof}
    First, note that in the case that $\Phi$ is completely positive, $\nu_m(\Phi) = 0$ for all $m$, in which case the limit clearly exists. Hence, assume that $\Phi$ is not completely positive, which in particular implies that $\nu_m(\Phi) > 0$ for at least some $m$ (and hence all further values of $m$). As $\nu_m(\Phi) \leq (m-1)(n-1)$ in general, the sequence $\nu_m(\Phi)/m$ resides in the interval $[0,n-1)$, and so $\alpha = \sup_m \nu_m(\Phi)/m < \infty$. We will show that in fact $\lim_{m\rightarrow \infty} \nu_m(\Phi)/m = \alpha$.
    
    Let $\epsilon > 0$. By definition of the supremum, there exists $m_0$ for which $\alpha - \nu_{m_0}(\Phi)/m_0 < \epsilon / 2$. Consider $m > m_0$, and let $k$ be a natural number for which $km_0 \leq m < (k+1)m_0$. It holds that
    \begin{equation}
        \frac{\nu_m(\Phi)}{m} > \frac{\nu_{km_0}(\Phi)}{(k+1)m_0} \geq \frac{k}{k+1} \frac{\nu_{m_0}(\Phi)}{m_0}.
    \end{equation}
    Hence, by choosing $k_0$ large enough, we may ensure that $\nu_m(\Phi)/m > v_{m_0}/m_0 - \epsilon/2$ for all $m \geq k_0m_0$. In particular, this implies that for all $m \geq k_0m_0$, it holds that
    \begin{equation}
    	\alpha - \nu_m(\Phi)/m < \alpha - \nu_{m_0}(\Phi)/m_0 + \epsilon/2 < \epsilon,
     \end{equation}and hence $\lim_{m \rightarrow \infty} \nu_m(\Phi)/m = \alpha$.
\end{proof}
\end{proposition}

We now start proving some less obvious bounds on $\nu_m(\Phi)$ and $\nu(\Phi)$. We start by showing that, when computing $\nu_m(\Phi^*)$, it suffices to only consider how many negative eigenvalues are produced by states $\rho \in \D(M_m,M_n)$ with the condition that $\tr_2(\rho) = I_m/m$.

\begin{lemma}\label{lemma:k_neg_eigenvalues}
    Suppose $\Phi \in \HP(M_n,M_n)$. The following are equivalent.
    \begin{enumerate}
        \item For all $\rho \in \D(M_m \otimes M_n)$, $(I_m \otimes \Phi)(\rho)$ has at most $k$ negative eigenvalues.
        \item For all $\rho \in \D(M_m \otimes M_n)$ with $\tr_2(\rho) = I_m/m$, $(I_m \otimes \Phi)(\rho)$  has at most $k$ negative eigenvalues.
    \end{enumerate}
\end{lemma}

\begin{proof}
    It is enough to show that statement~$2$ implies statement~$1$, so we assume throughout this proof that statement~$2$ holds. Let $\rho \in \D(M_m \otimes M_n)$, and first consider the case where the reduced density matrix $\sigma=\tr_2(\rho)$ has full rank. Letting $A = \sigma^{-\frac{1}{2}}/\sqrt{m}$, it holds that $\rho^\prime = (A^* \otimes I_n)\rho(A \otimes I_n)$ is a density matrix with $\tr_2(\rho^\prime) = I_m/m$. By statement 2, we know that $(I_m \otimes \Phi)(\rho^\prime)$ has at most $k$ negative eigenvalues. However, it holds that
    \begin{equation}
        (I_m \otimes \Phi)(\rho^\prime) = (I_m \otimes \Phi)( (A^* \otimes I_n)\rho(A \otimes I_n)) = (A^* \otimes I_n)(I_m \otimes \Phi)(\rho)(A \otimes I_n),
    \end{equation}
    and as $A$ is invertible, Sylvester's Law of Inertia tells us that $(I_m \otimes \Phi)(\rho)$ has the same number of negative eigenvalues as $(I_m \otimes \Phi)(\rho^\prime)$: at most $k$.
    
    Next, consider a general density matrix $\rho \in \D(M_m \otimes M_n)$, with not-necessarily-invertible partial trace. Denote $\rho_\lambda = (1-\lambda)\rho + \lambda I_m \otimes I_n/mn$. Then $\rho_\lambda$ is a density matrix for all $\lambda \in [0,1]$ and $\tr_2(\rho_\lambda)$ has full rank for all $\lambda \in (0,1]$. By the previous case, statement $2$ implies that $(I_m \otimes \Phi)(\rho_\lambda)$ has at most $k$ negative eigenvalues for all $\lambda > 0$. By the continuity of the $(k + 1)^{th}$ eigenvalue (when ordered in non-decreasing order), it follows that $(I_m \otimes \Phi)(\rho) = \lim_{\lambda \rightarrow 0}(I_m \otimes \Phi)(\rho_\lambda)$ has at most $k$ negative eigenvalues.
\end{proof}

We now present a lemma that provides a non-trivial upper bound on $\nu_m(\Phi)$ and $\nu(\Phi)$ for a wide variety of positive linear maps $\Phi$, based on the distance between $\Phi$ and the \emph{completely depolarizing map} $\Delta \in \HP(M_n,M_n)$ defined by
\[
    \Delta(X) = \tr(X)I/n \quad \text{for all} \quad X \in M_n.
\]



\begin{lemma}\label{lem:diamond_distance_bound}
    Suppose $\Phi \in \HP(M_n,M_n)$ and $x > 0$ is a scalar. Let $\displaystyle \ell(\Phi) = \min_{\rho \in \D(M_n)} \big\{ \tr(\Phi(\rho)) \big\}$. Then
    \[
        \nu_m(\Phi^*) \leq \left\lceil \frac{mn(x + d_m(x\Delta,\Phi) - \ell(\Phi))}{2x} \right\rceil - 1.
    \]
    In particular,
    \[
        \nu(\Phi^*) \leq \frac{n(x + d_{\diamond}(x\Delta,\Phi) - \ell(\Phi))}{2x}.
    \]
\end{lemma}

Before proving this lemma, we note that the quantity $\ell(\Phi)$ can be computed straightforwardly via semidefinite programming (see \cite{Wat18,BV04} for an introduction to the subject) and that $\ell(I_m \otimes \Phi) = \ell(\Phi)$ for all $m \geq 1$. The distance $d_m(x\Delta,\Phi)$ in general is difficult to compute (see \cite{JohCompute} for a randomized method that can produce lower bounds on it), but when $m \geq n$ it is just a diamond norm computation that can also be carried out efficiently via semidefinite programming \cite{Wat09,Wat13}. In particular, we have $d_m(x\Delta,\Phi) \leq d_\diamond(x\Delta,\Phi)$ for all $m$, with equality when $m \geq n$.

We also note that the scalar $x$ can be chosen freely in this lemma, and some choices of $x$ may be better than others. In practice, the optimal value of $x$ in the bound on $\nu(\Phi^*)$ can be found simply by computing that bound (via semidefinite programming) for several different values of $x$.

\begin{proof}[Proof of Lemma~\ref{lem:diamond_distance_bound}]
    Define $\Psi \in \HP(M_n,M_n)$ by $\Psi(X) = x\Delta(X) - \Phi(X)$, so that $\|\Psi\|_m = d_m(x\Delta,\Phi)$. By Lemma~\ref{lemma:k_neg_eigenvalues} we can assume that $\tr_2(\rho) = I_m/m$, so that
    \[
        (I_m \otimes \Phi)(\rho) = \frac{x}{mn}(I_m \otimes I_n) - (I_m \otimes \Psi)(\rho).
    \]
    
    We conclude that the sum of the positive eigenvalues of $(I_m \otimes \Psi)(\rho)$ is no larger than $(\|\Psi\|_m + (x - \ell(\Phi)))/2$ (since its trace norm is no larger than $\|\Psi\|_m$ and its trace is no larger than $x - \ell(\Phi)$). It follows that $(I_m \otimes \Psi)(\rho)$ must have strictly fewer than
    \[
        \frac{mn}{x}\left(\frac{\|\Psi\|_m + (x - \ell(\Phi))}{2}\right)
    \]
    eigenvalues that are strictly bigger than $x/(mn)$, so $(I_m \otimes \Phi)(\rho)$ can have no more than this many negative eigenvalues. Since the number of negative eigenvalues of $(I_m \otimes \Phi)(\rho)$ is an integer and this inequality is strict, we get exactly the bound described by the lemma.
\end{proof}

\section{Results for Certain Specific Positive Maps}\label{sec:specific_maps}

We now compute and bound the non-CP ratio (and the non-$m$-positive dimension) of some well-known positive linear maps from quantum information theory. We remind the reader that it is already known that if $T \in \HP(M_n,M_n)$ is the transpose map then $\nu_m(T) = (m-1)(n-1)$ and $\nu(T) = n-1$, which are the maximum values that $\nu_m$ and $\nu$ can take on positive maps.

\subsection{The Reduction Map}\label{sec:k_red_map}

Suppose $1 \leq k < n$ and consider the linear map $R_k \in \HP(M_n,M_n)$ defined by
\begin{align}\label{eq:kred_map}
    R_k(X) = k\tr(X)I_n - X.
\end{align}
We note that this map is known to be $k$-positive but not $(k+1)$-positive \cite{Tom85}. In the $k = 1$ case it is called the \emph{reduction map} \cite{CAG99,HH99}, and it is important in quantum information theory because it not only can be used to detect entanglement in quantum states (like all positive but not completely positive maps), but it also has the property that if $(I_n \otimes R_1)(\rho) \not\succeq O$ then $\rho$ is distillable (i.e., a pure Bell state can be distilled from many copies of $\rho$).

We now consider the problem of computing $\nu_m(R_k)$ and $\nu(R_k)$. That is, we answer the question of how large a subspace $\mathcal{S} \subseteq \C^m \otimes \C^n$ can be if it has the property that every state $\rho$ with support in $\mathcal{S}$ is such that $(I_m \otimes R_k)(\rho) \not\succeq O$ (in the $k = 1$ case, we say that such a state violates the \emph{reduction criterion}). Equivalently, since $R_k = R_k^*$, this answers the question of how many negative eigenvalues $(I_m \otimes R_k)(\rho)$ can have if $\rho \in \D(M_m \otimes M_n)$.

\begin{theorem}\label{thm:red_nminus1_neg_eigs}
    Suppose $1 \leq k \leq n$, $m \geq 1$, and $R_k \in \HP(M_n,M_n)$ is defined as in Equation~\eqref{eq:kred_map}. Then
    \[
        \nu_m(R_k) = \lceil m/k \rceil - 1 \quad \text{and} \quad \nu(R_k) = 1/k.
    \]
\end{theorem}

\begin{proof}
    The fact that $\nu(R_k) = 1/k$ follows immediately by the formula for $\nu_m(R_k)$, so we focus solely on the latter quantity.
    
    We start by showing that $\nu_m(R_k) \geq \lceil m/k \rceil - 1$ by constructing states $\rho \in \D(M_m \otimes M_n)$ for which $k\rho_1 \otimes I_n - \rho$ has $\lceil m/k \rceil -1$ negative eigenvalues. To construct such a state, let $\{\ket{u_1},\ket{u_2},\ldots,\ket{u_{\lceil m/k\rceil-1}}\}$ be any set consisting of $\lceil m/k\rceil-1$ mutually orthogonal maximally entangled pure states in $\C^m \otimes \C^n$. Then then state
    \[
        \rho := \frac{1}{\lceil m/k\rceil-1} \sum_{i=1}^{\lceil m/k\rceil-1} \ketbra{u_i}{u_i}
    \]
    satisfies
    \[
        (I_m \otimes R_k)(\rho) = \frac{k}{m}I_m \otimes I_n - \frac{1}{\lceil m/k\rceil-1}\sum_{i=1}^{\lceil m/k\rceil-1} \ketbra{u_i}{u_i},
    \]
    which has $\lceil m/k\rceil-1$ negative eigenvalues (since $1/(\lceil m/k\rceil-1) > k/m$).
    
    To see that $\nu_m(R_k) \leq \lceil m/k \rceil - 1$, we apply Lemma~\ref{lem:diamond_distance_bound} with $x = kn$. It is straightforward to show that $x\Delta - R_k$ is the identity map (so $d_m(x\Delta,R_k) = 1$) and $\tr(R_k(\rho)) = kn - 1$ for all $\rho \in \D(M_n)$, so (in the notation of the lemma) we have $\ell(R_k) = kn - 1$ and thus
    \[
        \nu_m(R_k^*) = \nu_m(R_k) \leq \left\lceil \frac{mn(kn + 1 - (kn-1))}{2kn} \right\rceil - 1 = \lceil m/k \rceil - 1,
    \]
    as claimed.
\end{proof}

It is worth mentioning that, although the maps $R_k$ are typically only considered when $k$ is an integer, these results still work even if it is not. For example, if $m = n$ and $k = n/(n-1)$ then $R_k$ is $\lfloor n/(n-1) \rfloor = 1$-positive and $(I_n \otimes R_k)(\rho)$ has at most $\lceil n/(n/(n-1)) \rceil - 1 = n-2$ negative eigenvalues.

\subsection{The Choi Map}\label{sec:choi_map}

The \emph{Choi map} \cite{Cho75b} is the positive map $\Phi_C \in \HP(M_3,M_3)$ that is defined by
\begin{align}\label{eq:choi_map}
	\Phi_{C}(X) \defeq \begin{bmatrix}
		x_{1,1} + x_{2,2} & -x_{1,2} & -x_{1,3} \\
		-x_{2,1} & x_{2,2} + x_{3,3} & -x_{2,3} \\
		-x_{3,1} & -x_{3,2} & x_{3,3} + x_{1,1}
	\end{bmatrix} \quad \text{for all} \quad X \in M_3.
\end{align}
This map is interesting for the fact that it was the first known (and is still the simplest known) example of a positive map that is \emph{indecomposable}, meaning that it is capable of detecting entanglement in some PPT states. Since $\Phi_C$ is positive, we know immediately that $\nu_m(\Phi_C) \leq (m-1)(n-1) = 2m-2$ so $\nu(\Phi_C) \leq 2$. We now provide a slightly better bound.

\begin{theorem}\label{thm:choi_map_bound}
    Suppose $m \geq 1$ and $\Phi_C \in \HP(M_3,M_3)$ is the Choi map defined in Equation~\eqref{eq:choi_map}. Then
    \[
        \nu_m(\Phi_C) = \nu_m(\Phi_C^*) \leq \left\lceil \frac{5m}{3} \right\rceil - 1 \quad \text{and} \quad \nu(\Phi_C) = \nu(\Phi_C^*) \leq 5/3.
    \]
\end{theorem}

\begin{proof}
    The fact that $\nu_m(\Phi_C) = \nu_m(\Phi_C^*)$ follows via a symmetry argument and the fact that $\Phi_C^*$ has the form
    \[
        \Phi_{C}^*(X) = \begin{bmatrix}
    		x_{1,1} + x_{3,3} & -x_{1,2} & -x_{1,3} \\
    		-x_{2,1} & x_{2,2} + x_{1,1} & -x_{2,3} \\
    		-x_{3,1} & -x_{3,2} & x_{3,3} + x_{2,2}
    	\end{bmatrix} \quad \text{for all} \quad X \in M_3.
    \]
    
    To compute the upper bound on $\nu_m(\Phi_C^*)$, we choose $x = 3$ in Lemma~\ref{lem:diamond_distance_bound}. Semidefinite programming then quickly shows that $d_\diamond(x\Delta,\Phi_C) = 7/3$. Similarly, for every $\rho \in \D(M_3)$ we have $\tr(\Phi_C(\rho)) = 2$, so $\ell(\Phi_C) = 2$. Plugging these details into Lemma~\ref{lem:diamond_distance_bound} gives
    \begin{align*}
        \nu_m(\Phi_C) & \leq \left\lceil \frac{mn(x + d_\diamond(x\Delta,\Phi_C) - \ell(\Phi_C))}{2x} \right\rceil - 1 = \left\lceil \frac{5m}{3} \right\rceil - 1.
    \end{align*}
    
    The fact that $\nu(\Phi_C^*) \leq 5/3$ then follows immediately from evaluating the limit in Definition~\ref{defn:non_cp_dim}.
\end{proof}

We note that we do not expect that the bounds on $\nu_m(\Phi_C)$ and $\nu(\Phi_C)$ provided by Theorem~\ref{thm:choi_map_bound} are tight. Numerics based on $10^6$ randomly-generated density matrices of various ranks (uniformly distributed according to Haar/Hilbert--Schmidt measure) in $\D(M_m \otimes M_n)$, for each value of $2 \leq m \leq 10$, suggest that $\nu_m(\Phi_C) \leq m-1$ and thus $\nu(\Phi_C) \leq 1$. In fact, even these bounds may not be tight---when $m = 10$, for example, we could not find a state $\rho \in \D(M_m \otimes M_n)$ for which $(I_m \otimes \Phi_C)(\rho)$ has more than $8$ negative eigenvalues. It thus may be the case that $\nu(\Phi_C) < 1$ (or it may just be the case that it's difficult to find the states $\rho$ for which $(I_m \otimes \Phi_C)(\rho)$ has many negative eigenvalues).

\subsection{The Breuer--Hall Map}

When $n \geq 4$ is an even integer, the \emph{Breuer--Hall map} $\Phi_B \in \HP(M_n,M_n)$ \cite{Bre06,Hal06} is defined by
\begin{align}\label{eq:BH_defn}
    \Phi_B(X) \defeq \Tr(X)I - X - UX^TU^*,
\end{align}
where $U$ is any fixed skew-symmetric unitary matrix (which is why $n$ must be even---such unitaries do not exist otherwise). For our purposes, we can assume that
\[
    U = \begin{bmatrix}
        0 & 1 & 0 & 0 & \cdots & 0 & 0 \\
        -1 & 0 & 0 & 0 & \cdots & 0 & 0 \\
        0 & 0 & 0 & 1 & \cdots & 0 & 0 \\
        0 & 0 & -1 & 0 & \cdots & 0 & 0 \\
        \vdots & \vdots & \vdots & \vdots & \ddots & \vdots & \vdots \\
        0 & 0 & 0 & 0 & \cdots & 0 & 1 \\
        0 & 0 & 0 & 0 & \cdots & -1 & 0
    \end{bmatrix},
\]
since every skew-symmetric unitary matrix is unitarily similar to this one, and thus $\nu_m(B)$ does not depend on which skew-symmetric unitary matrix is used.

We note that $\Phi_B = \Phi_B^*$ and $\Phi_B$ is known to be positive, so we can conclude immediately that $\nu_m(\Phi_B) \leq (m-1)(n-1)$ and $\nu(\Phi_B) \leq n-1$. We now present a better upper bound.

\begin{theorem}\label{thm:bh_map_nu}
    Suppose $n \geq 4$ is even, $m \geq 1$, and $\Phi_B \in \HP(M_n,M_n)$ is the Breuer--Hall map defined in Equation~\eqref{eq:BH_defn}. Then
    \[
        \nu_m(\Phi_B) \leq \frac{m}{2}(n+2) - 1 \quad \text{and} \quad \nu(\Phi_B) \leq n/2 + 1.
    \]
\end{theorem}

When $n = 4$, this bound is trivial since it just says that $\nu_m(\Phi_B) \leq 3m - 1$ and $\nu(\Phi_B) \leq 3$, but we already knew that $\nu_m(\Phi_B) \leq (m-1)(n-1) = 3m - 3$ and $\nu(\Phi_B) \leq 3$. However, this bound is non-trivial for all $n \geq 6$, beating the easy bounds $\nu_m(\Phi_B) \leq (m-1)(n-1)$ and $\nu(\Phi_B) \leq n-1$ by roughly a multiplicative factor of $2$ when $m$ and $n$ are large.

\begin{proof}[Proof of Theorem~\ref{thm:bh_map_nu}]
    We choose $x = 2n$ in Lemma~\ref{lem:diamond_distance_bound}. It is then straightforward to check that $x\Delta - \Phi_B$ is completely positive, since
    \[
        (x\Delta - \Phi_B)(X) = U(\tr(X)I + X^T)U^* + X
    \]
    is the sum of two completely positive maps. It follows that
    \[
        d_m(x\Delta,\Phi_B) = \|(x\Delta - \Phi_B)^*(I)\| = \|U(nI + I)U^* + I\| = n+2.
    \]
    Similarly, for every $\rho \in \D(M_n)$ we have $\Phi_B(\rho) = I - \rho - U\rho^T U^*$, which has trace $n - 2$, so $\ell(\Phi_B) = n - 2$. Plugging these details into Lemma~\ref{lem:diamond_distance_bound} gives
    \begin{align*}
        \nu_m(\Phi_B^*) & \leq \left\lceil \frac{mn(x + d_m(x\Delta,\Phi_B) - \ell(\Phi_B))}{2x} \right\rceil - 1 \\
        & = \left\lceil \frac{mn(2n + (n+2) - (n-2))}{4n} \right\rceil - 1 = \frac{m}{2}(n+2) - 1.
    \end{align*}
    
    The bound on $\nu(\Phi_B)$ then follows immediately from using this bound on $\nu_m(\Phi_B)$ and evaluating the limit in Definition~\ref{defn:non_cp_dim}.
\end{proof}

Similarly to the bounds for the Choi map, we do not expect that the bounds provided by Theorem~\ref{thm:bh_map_nu} are tight (in fact, we already noted that the bound on $\nu_m(\Phi_B)$ \emph{cannot} be tight when $n = 4$). Numerics based on $10^6$ randomly-generated density matrices of various ranks for each value of $2 \leq m \leq 10$ and even $4 \leq n \leq 10$ (uniformly distributed according to Haar/Hilbert--Schmidt measure) in $\D(M_m \otimes M_n)$ suggest that $\nu_m(\Phi_B) \leq m$ and thus $\nu(\Phi_B) \leq 1$.

\section{Relationship With other Measures of Entanglement Detection}\label{sec:relationship_other_measures}

Since we are proposing that the non-CP ratio of a positive map $\Phi$ should be interpreted as a measure of how good $\Phi$ is at detecting entanglement in quantum states, it is worth comparing it to other ways of quantifying and classifying this concept that have appeared in the literature.

Recall from \cite{LKCH00} that if $W_1,W_2 \in M_m \otimes M_n$ are entanglement witnesses then $W_1$ is said to be \emph{finer} than $W_2$ if $\tr(W_1\rho) < 0$ whenever $\tr(W_2\rho) < 0$, and in such a case we think of $W_1$ as a ``better'' entanglement witness than $W_2$. Furthermore, $W_1$ is said to be \emph{optimal} if there is no entanglement witness (other than $W_1$ itself and its multiples) that is finer than it. We now show that these concepts agree with the non-CP ratio of a positive linear map in a natural way.

\begin{theorem}
    Suppose $\Phi_1,\Phi_2 \in \HP(M_n,M_n)$ are positive but not completely positive. If $J(\Phi_1)$ is finer than $J(\Phi_2)$ then $\nu(\Phi_1) \geq \nu(\Phi_2)$.
\end{theorem}

\begin{proof}
    We recall from \cite[Lemma~2]{LKCH00} that $J(\Phi_1)$ is finer than $J(\Phi_2)$ if and only if there is a scalar $0 < c \in \R$ and a matrix $P \in \Pos(M_n \otimes M_n)$ such that $J(\Phi_1) = cJ(\Phi_2) - P$. It follows that there is a completely positive map $\Psi \in \HP(M_n,M_n)$ (with $J(\Psi) = P$) such that $\Phi_1 = c\Phi_2 - \Psi$. The fact that $\nu(\Phi_1) \geq \nu(\Phi_2)$ now follows immediately from the bullet-point list of properties of $\nu$ from earlier.
\end{proof}

On the other hand, it is not the case that every optimal entanglement witness $J(\Phi)$ has $\nu(\Phi)$ being as large as possible (i.e., equal to $n-1$). For example, the (Choi matrix of the) Choi map $\Phi_C \in \HP(M_3,M_3)$ is known to be optimal (and even extreme \cite{Ha13}, which is a stronger property), but we showed in Theorem~\ref{thm:choi_map_bound} that it has $\nu(\Phi_C) \leq 5/3 < n-1 = 2$.

\section{Schmidt Number and k-Positivity}\label{sec:schmidt_k_pos}

Recall that if $\Phi \in \HP(M_n,M_n)$ is $k$-positive then $\nu_m(\Phi) \leq (m-k)(n-k)$ and that the transpose map shows in the $k = 1$ case that equality can be attained for all $m$ and $n$. We have not yet demonstrated the existence of a map attaining equality for larger values of $k$; the only $k$-positive map we have considered was the map $R_k$ from Section~\ref{sec:k_red_map}, which had $\nu_m(R_k) = \lceil m/k \rceil - 1 \ll (m-k)(n-k)$ and $\nu(R_k) = 1/k \ll n-k$.

Nonetheless, $k$-positive maps attaining the $(m-k)(n-k)$ bound do exist for all $m$, $n$, and $k$. Before stating and proving this result formally, we recall that the \emph{Schmidt rank} of a pure state $\ket{v} \in \C^m \otimes \C^n$ is the least integer $k$ such that we can write $\ket{v} = \sum_{i=1}^k \gamma_k\ket{u_i} \otimes \ket{w_i}$, and a map $\Phi \in \HP(M_n,M_n)$ is $k$-positive if and only if $(I_m \otimes \Phi)(\ketbra{v}{v}) \succeq O$ for all $\ket{v} \in \C^m \otimes \C^n$ with Schmidt rank $\leq k$ \cite{TH00}. Furthermore, we say that the Choi matrix $J(\Phi)$ of a $k$-positive map $\Phi$ is \emph{$k$-block positive} \cite{SSZ09} (so that a matrix is an entanglement witness if and only if it is $1$-block positive and not positive semidefinite).

\begin{theorem}\label{thm:sch_rank_k_pos}
    Suppose $1 \leq k \leq n$ and $m \geq 1$. There exists a $k$-positive map $\Phi \in \HP(M_n,M_n)$ such that $\nu_m(\Phi) = (m-k)(n-k)$.
\end{theorem}

\begin{proof}
    Throughout this proof, we assume that $m \leq n$; if $m > n$ then we just swap $m$ and $n$ and the proof does not change substantially.
    
    Recall from \cite{CW08} that there exists a subspace $\mathcal{S} \subseteq \C^m \otimes \C^n$ with $\operatorname{dim}(\mathcal{S}) = (m-k)(n-k)$ and the property that it contains no states with Schmidt rank $\leq k$. If $P \in \Pos(M_m \otimes M_n)$ is the orthogonal projection onto $\mathcal{S}$, then (via the same argument from the proof of Theorem~\ref{thm:neg_evals_subspace_corr} in the $k=1$ case) there exists a scalar $c > 1$ such that $I - cP$ is $k$-block positive with $(m-k)(n-k)$ negative eigenvalues. It follows that $I - cP$ is the Choi matrix of some $k$-positive linear map $\Phi \in \HP(M_m,M_n)$ with the property that $I - cP = J(\Phi) = n(I_m \otimes \Phi)\big(\ketbra{\psi^+}{\psi^+}\big)$ has $(m-k)(n-k)$ negative eigenvalues. We can simply pad $\Phi$ with zeros to turn it into a $k$-positive map on $M_n$ producing the same number of negative eigenvalues, and then Theorem~\ref{thm:neg_evals_subspace_corr} tells us that $\nu_m(\Phi^*) = (m-k)(n-k)$.
\end{proof}

The proof of the above theorem can \emph{almost} be made constructive, with the one difficulty being that we do not know exactly how small we have to choose $c > 1$ in order to retain $k$-positivity of $\Phi$.

\begin{example}
    We now construct a $3$-positive linear map $\Phi \in \HP(M_5,M_5)$ with $\nu_5(\Phi) = (5-3)^2 = 4$, which is maximal in these dimensions. Set $\omega = e^{2i\pi/5}$ and define a completely positive map $\Omega \in \HP(M_5,M_5)$ by $\Omega(X) = \sum_{i=1}^4 A_iXA_i^*$, where
    \begin{align*}
        A_1 & = \frac{1}{2}\begin{bmatrix}
            0 & 1 & 0 & 0 & 0 \\
            0 & 0 & 1 & 0 & 0 \\
            0 & 0 & 0 & 1 & 0 \\
            0 & 0 & 0 & 0 & 1 \\
            0 & 0 & 0 & 0 & 0 \\
        \end{bmatrix}, & A_2 & = \frac{1}{2}\begin{bmatrix}
            0 & 0 & 0 & 0 & 0 \\
            1 & 0 & 0 & 0 & 0 \\
            0 & 1 & 0 & 0 & 0 \\
            0 & 0 & 1 & 0 & 0 \\
            0 & 0 & 0 & 1 & 0 \\
        \end{bmatrix},\\
        A_3 & = \frac{1}{\sqrt{5}}\begin{bmatrix}
            1 & 0 & 0 & 0 & 0 \\
            0 & 1 & 0 & 0 & 0 \\
            0 & 0 & 1 & 0 & 0 \\
            0 & 0 & 0 & 1 & 0 \\
            0 & 0 & 0 & 0 & 1 \\
        \end{bmatrix}, & A_4 & = \frac{1}{\sqrt{5}}\begin{bmatrix}
            1 & 0 & 0 & 0 & 0 \\
            0 & \omega & 0 & 0 & 0 \\
            0 & 0 & \omega^2 & 0 & 0 \\
            0 & 0 & 0 & \omega^3 & 0 \\
            0 & 0 & 0 & 0 & \omega^4 \\
        \end{bmatrix}.
    \end{align*}
    
    If $c$ is any scalar then it is straightforward to show that the map $\Phi_c \in \HP(M_5,M_5)$ defined by $\Phi_c(X) = \tr(X)I - c\Omega(X)$ has Choi matrix
    \[
        J(\Phi_c) = I - cP,
    \]
    where $P$ is the orthogonal projection onto the $4$-dimensional subspace
    \[
        \mathcal{S} := \operatorname{span}\{\operatorname{vec}(A_i) : i = 1, 2, 3, 4 \} \subset \C^{25}
    \]
    and $\operatorname{vec}(A_i)$ refers to the \emph{vectorization} of $A_i$ (i.e., the vector in $\C^{25}$ that is obtained by reading $A_i$ row-by-row).
    
    It follows that if $c > 1$ then $J(\Phi_c) = 5(I_5 \otimes \Phi_c)\big(\ketbra{\psi^+}{\psi^+}\big)$ has $4$ negative eigenvalues, so $\nu(\Phi_c^*) = 4$. Furthermore, we can see that every state in $\mathcal{S}$ has Schmidt rank strictly bigger than $3$ via the argument of \cite{CW08}: every linear combination of $A_1$, $A_2$, $A_3$, $A_4$ has at least $4$ non-zero entries on its leading diagonal and thus has rank $\geq 4$, and thus every linear combination of $\operatorname{vec}(A_1)$, $\operatorname{vec}(A_2)$, $\operatorname{vec}(A_3)$, $\operatorname{vec}(A_4)$ has Schmidt rank $\geq 4$. It follows that there exists some $c_* > 1$ such that $\Phi_c$ (and thus $\Phi_c^*$) is $3$-positive whenever $1 < c \leq c_*$.
\end{example}

\section{A Multipartite NPT Subspace}\label{sec:multipartite_ppt}

Recall that if $T$ is the transpose map acting on $M_n$ then $\nu_m(T) = (m-1)(n-1)$, which means two (essentially equivalent) things: (a) the largest subspace $\mathcal{S} \subseteq \C^m \otimes \C^n$ with $(I_m \otimes T)(\rho) \not\succeq O$ whenever $\supp(\rho) \in \mathcal{S}$ has $\dim(\mathcal{S}) = (m-1)(n-1)$, and (b) the maximal number of negative eigenvalues that $(I_m \otimes T)(\rho)$ can have is $(m-1)(n-1)$. Here we generalize these statements to the multipartite case.

That is, we now consider the problem of how large a subspace $\S \subseteq \C^{d_1} \otimes \C^{d_2} \otimes \cdots \otimes \C^{d_p}$ can be if it has the property that every $\rho$ supported on $\S$ has non-positive partial transpose across at least one cut (we call such a subspace an \emph{NPT subspace}). Since any such subspace must be entangled, it cannot have dimension larger than the Parthasarathy bound \cite{Par04} of $d_1d_2\cdots d_p - d_1-d_2-\cdots-d_p + p-1$. As we have noted before, this bound was shown to be attainable by an NPT subspace in the bipartite case (in which case the bound simplifies to $d_1d_2 - d_1 - d_2 + 1 = (d_1-1)(d_2-1)$) in \cite{Joh13}, and partial progress on this problem was made in the multipartite case in \cite{SAS14}. We now solve this problem by showing that there is an NPT subspace that attains the Parthasarathy bound in the multipartite case as well.

To this end, define
 \begin{align}\label{eq:J_set}
    J \defeq \{0,1,2,\ldots, d_1+d_2+\cdots+d_p-p\},
 \end{align}
 and for each $s \in J$ let
\begin{align*}
    I_s \defeq \{(i_1,i_2,\ldots,i_p) : i_1+i_2+\cdots+i_p=s\}.
\end{align*}
We then consider the subspace $\mathcal{S} \subset \C^{d_1} \otimes \C^{d_2} \otimes \cdots \otimes \C^{d_p}$ defined by
\begin{align}\label{eq:npt_multi_subspace}
    \mathcal{S} = \left\{ \mathbf{v} \in \C^{d_1} \otimes \C^{d_2} \otimes \cdots \otimes \C^{d_p} : \sum_{i \in I_s} v_i = 0 \ \ \text{for all} \ \ s \in J \right\},
\end{align}
where we index the entries of the tensor $\ket{v}$ via tuples in the obvious way (i.e., if $i = (i_1,i_2,\ldots,i_p)$ then $v_i = (\bra{i_1}\otimes\bra{i_2}\otimes\cdots\otimes\bra{i_p})\mathbf{v}$). This subspace places $1$ linear restriction on $\ket{v}$ for each of the $|J| = d_1+d_2+\cdots+d_p-p+1$ possible values of $s$, so it has dimension $\dim(\mathcal{S}) = d_1d_2\cdots d_p - d_1-d_2-\cdots-d_p+p-1$, as claimed. Furthermore, we have the following result.

\begin{theorem}\label{thm:npt_subspace_multi}
    Let $\S$ be the subspace defined in Equation~\eqref{eq:npt_multi_subspace}. Every state $\rho \in \D(M_{d_1} \otimes M_{d_2} \otimes \cdots \otimes M_{d_p})$ with $\supp(\rho) \subseteq \S$ has the property that its partial transpose on the $j$-th subsystem is non-positive for some $1 \leq j < p$. In particular, $\S$ is an NPT subspace.
\end{theorem}

We note that this theorem is optimal in two different ways---not only is its dimension as large as the dimension of an NPT subspace could possibly be, but the fact that it only considers $p-1$ single-system partial transpositions out of the $2^p$ total possible partial transpositions also cannot be improved upon. For example, if we were to only consider the partial transposition across one of the first $p-2$ subsystems then no such subspace of this dimension exists, since the largest entangled subspace in $\C^{d_1} \otimes \cdots\otimes \C^{d_{p-2}}  \otimes (\C^{d_{p-1}d_p})$ has dimension
\[
    d_1d_2\cdots d_p - d_1-d_2-\cdots-d_{p-2}-d_{p-1}d_p+p-2 < d_1d_2\cdots d_p - d_1-d_2-\cdots-d_p+p-1.
\]
For example, if we just require that the first single-system partial transposition on $M_2 \otimes M_2 \otimes M_2$ be non-positive, we can make the identification $M_2 \otimes M_2 \otimes M_2 \cong M_2 \otimes M_4$, and in that space the largest entangled subspace has dimension $(2-1)(4-1) = 3$, but our method constructs an NPT subspace of dimension $4$.

Before proceeding, we also note that the proof of this result is extremely technical, but morally is very similar to the proof of \cite[Theorem~1]{Joh13}, which is the analogous statement in the bipartite case. We thus recommend that the reader makes themselves familiar with the technicalities of that proof before delving into this one.

\begin{proof}[Proof of Theorem~\ref{thm:npt_subspace_multi}]
    Let
    \begin{align*}
        \rho = \sum_{a} p_a \ketbra{v^{(a)}}{v^{(a)}}\in \D(\C^{d_1}\otimes\dots\otimes \C^{d_p})
    \end{align*}
    be any density matrix with support contained in the subspace $\S$ defined in Equation~\eqref{eq:npt_multi_subspace} (i.e. $\ket{v^{(a)}} \in \mathcal{S}$ for all $a$). Our goal is to show that $\rho$ is NPT across some bipartite cut of $M_{d_1} \otimes \dots \otimes M_{d_p}$.
    
    Let $s \in J$ (where $J$ is as in Equation~\eqref{eq:J_set}) be the smallest integer such that $v^{(a)}_i \neq 0$ for some $a$ and $i \in I_s$. Let $i^\prime \neq j^\prime \in I_s$ be any indices for which
    \begin{align}\label{eq:choose_submatrix}
        \sum_{a} p_a \overline{v^{(a)}_{i'}}v^{(a)}_{j'} \neq 0
    \end{align}
    (we defer showing the existence of such indices until the end of the proof).
    
    Write $i^\prime=(i^\prime_1,i^\prime_2,\ldots,i^\prime_p)$ and $j^\prime=(j^\prime_1,j^\prime_2,\ldots,j^\prime_p)$. Since $\sum_{k=1}^p i^\prime_k=\sum_{k=1}^p j^\prime_k=s$ and $i^\prime \neq j^\prime$, there exist indices $1 \leq k^\prime \neq k^{\prime\prime} \leq p$ such that $i^\prime_{k^\prime} > j^\prime_{k^\prime}$ and $i^\prime_{k^{\prime\prime}} < j^\prime_{k^{\prime\prime}}$.
    
    Let $\sigma : \{1,2,\ldots,p\} \rightarrow \{1,2,\ldots,p\}$ be a permutation for which
    \begin{align*}
        \sigma(k^\prime) = 1 \quad \text{and} \quad \sigma(k^{\prime\prime}) = 2.
    \end{align*}
    Such a permutation $\sigma$ has the property that
    \begin{align*}
        (i^\prime_{\sigma^{-1}(2)},i^\prime_{\sigma^{-1}(3)},\ldots,i^\prime_{\sigma^{-1}(p)}) < (j^\prime_{\sigma^{-1}(2)},j^\prime_{\sigma^{-1}(3)},\ldots,j^\prime_{\sigma^{-1}(p)})
    \end{align*}
    when these tuples are viewed as $d_{\sigma^{-1}(2)},d_{\sigma^{-1}(3)},\ldots, d_{\sigma^{-1}(p)}$-ary strings---a property that follows from the fact that 
    \begin{align*}
        i^\prime_{\sigma^{-1}(2)}=i^\prime_{{k^{\prime\prime}}} < j^\prime_{{k^{\prime\prime}}}=j^\prime_{\sigma^{-1}(2)}.
    \end{align*}
    Let $W$ be the (unitary) swap matrix that permutes the tensor factors of $\C^{d_1} \otimes \C^{d_2} \otimes \cdots \otimes \C^{d_p}$ according to $\sigma$, i.e. on basis vectors
    \begin{align*}
    \ket{i}=\ket{i_1}\otimes\dots\otimes \ket{i_p} \in \C^{d_1} \otimes \dots \otimes \C^{d_p},
    \end{align*}
    $W$ acts as
    \begin{align*}
        W\ket{i}=\ket{i_{\sigma^{-1}(1)}}\otimes\dots \otimes \ket{i_{\sigma^{-1}(p)}} \in \C^{d_{\sigma^{-1}(1)}} \otimes \dots \otimes \C^{d_{\sigma^{-1}(p)}}.
        \end{align*}
        Let
    \begin{align*}
        \mu=W\rho W^*,
    \end{align*}
    so $\mu$ is NPT across some bipartite cut if and only if $\rho$ is. We now show that $\mu$ is NPT across the cut that separates $M_{d_{k^\prime}}$ from the other subsystems (i.e., if we let $T_{k^\prime}$ denote the partial transpose on the $k^\prime$-th subsystem then $T_{k^\prime}(\mu) \not\succeq O$).
    
    Let
    \begin{align*}
        D : \bigotimes_{k = 1}^p \C^{d_{\sigma^{-1}(k)}} \rightarrow \Lin\left(\bigotimes_{k=2}^p \C^{d_{\sigma^{-1}(k)}},\C^{d_{k^\prime}}\right)
    \end{align*}
    be the linear map defined on standard basis vectors by
    \begin{align*}
        D\left(\bigotimes_{k=1}^p \ket{i_{\sigma^{-1}(k)}}\right) =\ket{i_{k^\prime}}\bigotimes_{k=2}^p\bra{i_{\sigma^{-1}(k)}}.
    \end{align*}
    Then for each $a$, the $2 \times 2$ submatrix of $D(W\ket{v^{(a)}})$ corresponding to the rows $\ket{i^\prime_{k^\prime}}$, $\ket{j^\prime_{k^\prime}}$ and columns $\ket{i^\prime_{\sigma^{-1}(2)},i^\prime_{\sigma^{-1}(3)},\ldots, i^\prime_{\sigma^{-1}(p)}}$, $\ket{j^\prime_{\sigma^{-1}(2)},j^\prime_{\sigma^{-1}(3)},\ldots, j^\prime_{\sigma^{-1}(p)}}$ must take the form
    \begin{align*}
        \begin{bmatrix}
            0 & v^{(a)}_{j^\prime} \\
            v^{(a)}_{i^\prime} & *
        \end{bmatrix},
    \end{align*}
    where the asterisk ($*$) is a value we don't care about. The zero in the upper-left corner follows from the fact that this entry is equal to $v^{(a)}_{\ell}$ for
    \begin{align*}
        \ell=(i^\prime_1,\dots, i^\prime_{k^\prime-1},j^\prime_{k^\prime},i^\prime_{k^\prime+1},\dots, i^\prime_p),
    \end{align*}
    and
    \begin{align*}
        j^\prime_{k^\prime}+\sum_{\substack{k = 1 \\ k \neq k^\prime}}^p i^\prime_k < \sum_{k = 1}^p i^\prime_k = s,
    \end{align*}
    so $v^{(a)}_\ell = 0$ by our construction. A direct calculation now reveals that the $2 \times 2$ principal submatrix of $T_{k^\prime}(W\ketbra{v^{(a)}}{v^{(a)}}W^*)$ (where we recall that $T_{k^\prime}$ denotes the partial transpose on the $k^\prime$-th subsystem) corresponding to rows and columns $\ket{i^\prime_{\sigma^{-1}(1)},i^\prime_{\sigma^{-1}(2)},\ldots, i^\prime_{\sigma^{-1}(p)}}$ and $\ket{j^\prime_{\sigma^{-1}(1)},j^\prime_{\sigma^{-1}(2)},\ldots, j^\prime_{\sigma^{-1}(p)}}$ is
    \begin{align*}
    	\begin{bmatrix}0 & v^{(a)}_{i^\prime}\overline{v^{(a)}_{j^\prime}} \\ \overline{v^{(a)}_{i^\prime}}v^{(a)}_{j^\prime} & *\end{bmatrix},
    \end{align*}
    so the same principal submatrix of $T_{k^\prime}(\mu) = \sum_a p_a T_{k^\prime}(W\ketbra{v^{(a)}}{v^{(a)}}W^*)$ is
    \begin{align*}
    	\sum_a p_a \begin{bmatrix}0 & v^{(a)}_{i^\prime}\overline{v^{(a)}_{j^\prime}} \\ \overline{v^{(a)}_{i^\prime}}v^{(a)}_{j^\prime} & *\end{bmatrix}.
    \end{align*}
    Since the determinant of this principal submatrix is
    \begin{align*}
    	-\left| \sum_a p_a\overline{v^{(a)}_{i^\prime}}v^{(a)}_{j^\prime} \right|^2,
    \end{align*}
    and we specifically chose this quantity to be non-zero in Equation~\eqref{eq:choose_submatrix}, we know that this determinant is strictly negative. It follows that this $2 \times 2$ principal submatrix of $T_{k^\prime}(\mu)$ has a negative eigenvalue, so $T_{k^\prime}(\mu)$ itself has a negative eigenvalue as well, as desired.
    
    We conclude by showing that there must exist indices $i^\prime \neq j^\prime \in I_s$ so as to satisfy Equation~\eqref{eq:choose_submatrix}:
    \begin{align*}
        \sum_{a} p_a \overline{v^{(a)}_{i^\prime}}v^{(a)}_{j^\prime} \neq 0.
    \end{align*}
    Indeed, assume towards contradiction that
    \begin{align}\label{eq:contr_assumpt_ppt}
        \sum_{a} p_a \overline{v^{(a)}_{i}}v^{(a)}_{j} = 0
    \end{align}
    for all $i \neq j \in I_s$. Let $i^\prime \in I_s$ be any index such that $v^{(a)}_{i^\prime} \neq 0$ for some $a$, and let $j^{\prime\prime} \neq i^\prime \in I_s$ be any other index. The fact that $\ket{v^{(a)}} \in \S$ tells us that
    \begin{align*}
        v^{(a)}_{j^{\prime\prime}} = -\sum_{\substack{i \in I_s \\ i \neq j^{\prime\prime}}} v^{(a)}_{i},
    \end{align*}
    which implies
    \begin{align*}
        0 = \sum_{a} p_a \overline{v^{(a)}_{i^\prime}}v^{(a)}_{j^{\prime\prime}} = \sum_{a} p_a \overline{v_{i^\prime}^{(a)}} \left(-\sum_{\substack{i \in I_s\\ i \neq j^{\prime\prime}}} v^{(a)}_{i}\right) = -\sum_{a} p_a \big|v^{(a)}_{i^\prime}\big|^2 < 0,
    \end{align*}
    where we note that we used the assumption~\eqref{eq:contr_assumpt_ppt} in both the first equality and the final equality above. Since $0 < 0$ is an obvious contradiction, we conclude that such indices $i^\prime \neq j^\prime \in I_s$ do exist, which completes the proof.
\end{proof}

\subsection{Eigenvalues of Multipartite Decomposable Entanglement Witnesses}\label{sec:multipartite_ppt_ews}

Just like there is a natural correspondence between the dimension of an NPT subspace and the number of negative eigenvalues that the partial transpose can produce in the bipartite case, there is also such a correspondence in the multipartite case. In particular, since the maximum dimension of a NPT subspace is $d_1d_2\cdots d_p - d_1-d_2-\cdots-d_p+p-1$, we can conclude (via the exact same argument used in the proof of Theorem~\ref{thm:neg_evals_subspace_corr}) that every multipartite decomposable entanglement witness (i.e., sum of partial transposes of different combinations of subsystems of positive semidefinite matrices) has at most that many negative eigenvalues.

Furthermore, Theorem~\ref{thm:npt_subspace_multi} shows that not only can decomposable entanglement witnesses not have more than $d_1d_2\cdots d_p - d_1-d_2-\cdots-d_p+p-1$ negative eigenvalues, but this number of negative eigenvalues can even be attained if we only sum up $p-1$ single-system partial transpositions of positive semidefinite matrices (rather than $p(p-1)/2$ such matrices like is required to construct general decomposable entanglement witnesses).

For example, if we define
\begin{align*}
    \ket{v_1} & = \ket{001} + \ket{010} + 2\ket{111}, & \ket{v_2} & = \ket{001} + \ket{100} + 2\ket{111}, \\
    \ket{w_1} & = \ket{110} + \ket{101} + 2\ket{000}, & \ket{w_2} & = \ket{110} + \ket{011} + 2\ket{000}, \quad \text{and} \\
    X_1 & = \ketbra{v_1}{v_1} + \ketbra{w_1}{w_1} & X_2 & = \ketbra{v_2}{v_2} + \ketbra{w_2}{w_2},
\end{align*}
then $X_1,X_2 \in \Pos(M_2 \otimes M_2 \otimes M_2)$ and
\begin{align*}
    (T \otimes I \otimes I)(X_1) + (I \otimes T \otimes I)(X_2) = \left[\begin{array}{c|c}\begin{array}{cc|cc}8 & \cdot & \cdot & \cdot \\
     \cdot & 2 & 3 & \cdot \\\hline
     \cdot & 3 & 1 & \cdot \\
     \cdot & \cdot & \cdot & 1
     \end{array} & \begin{array}{cc|cc}\cdot & \cdot & \cdot & \cdot \\
     3 & \cdot & \cdot & \cdot \\\hline
     4 & \cdot & \cdot & \cdot \\
     \cdot & 4 & 3 & \cdot\end{array}\\[-0.2cm] \\\hline\\[-0.2cm]
    \begin{array}{cc|cc}\cdot & 3 & 4 & \cdot \\
     \cdot & \cdot & \cdot & 4 \\\hline
     \cdot & \cdot & \cdot & 3 \\
     \cdot & \cdot & \cdot & \cdot\end{array} & \begin{array}{cc|cc}1 & \cdot & \cdot & \cdot \\
     \cdot & 1 & 3 & \cdot \\\hline
     \cdot & 3 & 2 & \cdot \\
     \cdot & \cdot & \cdot & 8\end{array} 
     \end{array}\right],
\end{align*}
is a decomposable entanglement witness with eigenvalues $8$, $8$, $8$, $8$, $-1$, $-1$, $-3$, and $-3$ (counting multiplicity). In particular, this is a decomposable entanglement witness with $4$ negative eigenvalues, which is maximal for a $3$-qubit entanglement witness.

In general, a multipartite decomposable entanglement witness with the maximal number of negative eigenvalues can be constructed via the same approach that was used in the bipartite case in \cite{Joh13}: we let $P \in \Pos(M_{d_1} \otimes M_{d_2} \otimes \cdots \otimes M_{d_p})$ be the orthogonal projection onto the subspace $\S$ from Equation~\eqref{eq:npt_multi_subspace} and $T_k$ denote the partial transpose on the $k$-th subsystem, and consider the following SDP in the variable $c$ and matrix variables $X_1,\ldots,X_{p-1}$:
\begin{align*}
    \begin{matrix}
        \begin{tabular}{r l}
        \text{maximize:} & $c$ \\
        \text{subject to:} & $T_1(X_1) + \cdots + T_{p-1}(X_{p-1}) \leq I - cP$ \\
        \ & $X_1,\ldots,X_{p-1} \succeq O$ \\
        \end{tabular}
    \end{matrix}
\end{align*}
The optimal value of this semidefinite program is strictly bigger than $1$, so $T_1(X_1) + \cdots + T_{p-1}(X_{p-1})$ (which is a decomposable entanglement witness) has exactly $\rank(P) = \dim(\S) = d_1d_2\cdots d_p - d_1-d_2-\cdots-d_p+p-1$ negative eigenvalues.

\section{Conclusions and Open Questions}\label{sec:conclusions}

In this work, we introduced the non-$m$-positive and non-completely-positive ratio of a Hermiticity-preserving linear map. We established the basic properties of this quantity, and motivated it as a measure of how good a positive (or $k$-positive) linear map is at detecting entanglement in quantum states. We then developed some methods of bounding these quantities and we applied our methods to well-known maps like the reduction map and its generalizations, the Choi map, and the Breuer--Hall map.

Some questions arising from this work that remain unanswered include:

\begin{itemize}
    \item What are the true values of $\nu(\Phi_B)$ and $\nu(\Phi_C)$, where $\Phi_B$ and $\Phi_C$ are the Breuer--Hall and Choi maps, respectively? We provided some bounds in Theorems~\ref{thm:choi_map_bound} and~\ref{thm:bh_map_nu}, but numerics suggest that they are not tight.
    
    \item We showed that if $\Phi \in \HP(M_n,M_n)$ is not $k$-positive then $\nu(\Phi) \geq 1/k$. Furthermore, the generalized reduction maps from Section~\ref{sec:k_red_map} show that, for any $\varepsilon > 0$, there is a non-$k$-positive map $\Phi$ with $\nu(\Phi) = 1/k + \varepsilon$ (we can choose $\Phi = R_{k-\varepsilon}$, for example). What about the $\varepsilon = 0$ case? Is $\nu(\Phi) = 1/k$ possible, or is it the case that $\nu(\Phi) > 1/k$?
    
    \item In Theorem~\ref{thm:npt_subspace_multi}, we computed the dimension of the largest subspace of $\mathcal{S} \subseteq \C^{d_1} \otimes \cdots \otimes \C^{d_p}$ with the property that, whenever $\rho$ is supported on $\mathcal{S}$, at least one of its single-party partial transpositions is non-positive. What is the largest such dimension if we require instead that at least $k \geq 2$ of its single-party partial transpositions are non-positive?
    
    \item In Theorem~\ref{thm:sch_rank_k_pos}, we showed that for all $1 \leq k \leq n$ and $m \geq 1$ there exists a $k$-positive map $\Phi \in \HP(M_n,M_n)$ such that $\nu_m(\Phi) = (m-k)(n-k)$. Is it possible to find a \emph{single} $k$-positive map $\Phi \in \HP(M_n,M_n)$ such that $\nu_m(\Phi) = (m-k)(n-k)$ for \emph{all} $m \geq 1$? In the $k = 1$ case, the transpose map does the job, but it is not immediately clear what happens when $k > 1$.
\end{itemize}

\noindent \textbf{Acknowledgements.} We thank John Watrous and Jamie Sikora for helpful discussions. N.J.\ was supported by NSERC Discovery Grant number RGPIN-2016-04003.

\printbibliography
\end{document}